\documentclass[aps,pra, twocolumn, nofootinbib]{revtex4-1}

\usepackage[T1]{fontenc }
\usepackage[utf8]{inputenc}

\usepackage{float}
\usepackage{graphicx}
\usepackage{epsfig,epstopdf}
\usepackage[table]{xcolor}
\usepackage{color}
\definecolor{marek}{rgb}{.5,.5,.1}

\usepackage{soul}

\usepackage{times}

\usepackage{amsmath}
\usepackage{bbm}
\usepackage{amsfonts}
\usepackage{amssymb}
\usepackage{amsthm}
\usepackage{mathtools}

\usepackage{algorithmic}

\usepackage{algorithm}

\usepackage{complexity}
\newclass{\sharpP}{\#P}
\newclass{\classP}{P}

\usepackage{enumerate}

\usepackage{enumitem}

\newtheorem{theorem}{Theorem}

\newtheorem{lemma}[theorem]{Lemma}

\definecolor{DarkGray}{gray}{0.25} 
\definecolor{MidGray}{gray}{0.38} 
\definecolor{NeutralGray}{gray}{0.5}
\definecolor{LightGray}{gray}{0.7}
\definecolor{lightGray}{gray}{0.85}
\definecolor{DarkRed}{rgb}{0.7,0,0}
\definecolor{DarkBlue}{rgb}{0,0,0.5}
\definecolor{SteelBlue}{rgb}{0,0.4,0.6}
\definecolor{Orange}{rgb}{0.7,0.5,0}
\definecolor{Violette}{rgb}{0.5,0,0.5}
\definecolor{Sand}{rgb}{0.84,0.8,0.55}
\definecolor{niceblue}{rgb}{0.33,0.5,0.8}
\definecolor{OliveGreen}{RGB}{0,102,102}
\definecolor{NiceGreen}{RGB}{0,153,72}
\definecolor{LightGreen}{RGB}{0,200,72}
\definecolor{newblue}{RGB}{40,210,251}
\definecolor{lightblue}{RGB}{179,231,251}
\definecolor{steelblue}{RGB}{70,130,180}
\definecolor{cred}{RGB}{179,28,28}
\definecolor{applegreen}{rgb}{0.55, 0.71, 0.0}
\usepackage{tcolorbox}

\tcbset{
    coltitle=black,fonttitle=\bfseries,
    colback=orange!7!white,colframe=gray!25!white,width=0.95\textwidth,valign=center,
    before=\par\bigskip\centering,after=\par
    }

\newcommand{\ket}[1]{|{#1}\rangle}

\newcommand{\bra}[1]{\langle{#1}|}
\newcommand{\braket}[2]{\left\langle{#1}\middle|{#2}\right\rangle}

\newcommand{\ketbra}[2]{\ket{#1} \!\! \bra{#2}}

\newcommand{\mb}[1]{\mathbb{#1}}

\newcommand{\p}{\textsf{P}}

\newcommand{\np}{\textsf{NP}}

\newtheorem{problem}{Problem}

\usepackage[breaklinks=true]{hyperref}

\hypersetup{%
    draft, 
    colorlinks=true, linktocpage=true, pdfstartpage=3, pdfstartview=FitV,%
    colorlinks=true, 
    linktocpage=false, pdfstartpage=3, 
    breaklinks=true, pdfpagemode=UseNone, 
    plainpages=false, bookmarksnumbered, 
    hypertexnames=true, 
    urlcolor=webbrown, linkcolor=DarkBlue, citecolor=NiceGreen, 
}

\definecolor{dominik}{RGB}{237,16,118}

\newcommand{\fu}{Dahlem Center for Complex Quantum Systems, Freie Universit{\"a}t Berlin, 14195 Berlin, Germany}

\begin{document}

\title{Contracting projected entangled pair states is average-case hard} 
\author{
Jonas Haferkamp, Dominik Hangleiter, Jens Eisert, Marek Gluza
}
\affiliation{\fu}


\begin{abstract}
An accurate calculation of the properties of quantum many-body systems is one of the most important yet intricate 
challenges of modern physics and computer science.
In recent years, the tensor network ansatz has established itself as one of the most promising approaches enabling striking efficiency of simulating static properties of one-dimensional systems and abounding numerical applications in condensed matter theory.
In higher dimensions, however, a connection to the field of computational complexity theory has shown that the accurate normalization of the two-dimensional tensor networks called projected entangled pair states (PEPS) is \sharpP-complete. Therefore, an efficient algorithm for PEPS contraction would allow to solve exceedingly difficult combinatorial counting problems, which is considered highly unlikely.
Due to the importance of understanding two- and three-dimensional systems the question currently remains: Are the known constructions typical of states relevant for quantum many-body systems?
In this work, we show that an accurate evaluation of normalization or expectation values of PEPS is as hard to compute for typical instances as for special configurations of highest computational hardness.
We discuss the structural property of average-case hardness in relation to the current research on efficient algorithms attempting tensor network contraction, 
hinting at a wealth of possible further insights into the average-case hardness of important problems in quantum many-body theory.
\end{abstract}


\maketitle
Determining the properties of quantum many-body systems is of paramount importance in our efforts to understand conductance and thermodynamics of solid-state materials \cite{altland2010condensed,nazarov2009quantum}, designing new sensors and devising novel quantum technologies \cite{acin2017european}, inferring nuclear processes in stars or the early universe \cite{weber2005strange,arsene2005quark}. 
However, oftentimes it is not possible to find degrees of freedom enabling a concise description of a given system in terms of an effective model featuring essentially no interactions.
In such a case there is usually no easy way out but to calculate numerically observables of interest from a Hamiltonian description \cite{georges1996dynamical,gull2011continuous,metzner2012functional,schollwock2005density,schollwock2011density,alet2018many,degrand2006lattice}.
Here, however, we face a particular challenge namely that the state space of quantum many-body systems demands a number of parameters that grows exponentially with the amount of constituents of the system.
If so, even storing the state of the system on a computer becomes impossible and hence one seeks for efficient variational families of states. 
Tensor networks are a prime example of such an ansatz class 
\cite{fannes1992finitely,verstraete2008matrix,schollwock2011density,orus2014practical,bridgeman2017hand,AreaReview}.
Despite their spectacular success in one-dimension \cite{verstraete2004matrix,gobert2005real,verstraete2006matrix,manmana2007strongly,trotzky2012probing,fukuhara2013quantum,PhysRevB.81.064439,bardarson2012unbounded,karrasch2012finite,PhysRevLett.106.217206,ronzheimer2013expansion,arad2017rigorous} as so-called matrix-product states \cite{perez2006matrix,verstraete2006matrix,verstraete2008matrix}, the most natural tensor network ansatz in two-dimensions, called projected entangled pair states (PEPS) \cite{verstraete2006criticality}, turned out to be burdened by a peculiar difficulty:
even to calculate the normalization of PEPS is computationally intractable as has been shown by \citet{schuch_hardness_2007}.

More precisely, the normalization or evaluation of a local expectation value within the PEPS ansatz class is a computational task which is complete for the complexity class \sharpP, i.e., is as hard as any other problem in this class \cite{arora2009computational,nielsen2002quantum,aaronson2013quantum}.
A paradigmatic \sharpP\, problem consists in counting the solutions of the traveling salesman problem, which is an optimization problem, complete for the class \np. 
Intuitively, counting the solutions to a hard problem can only be harder. 
Within the current state of knowledge in computer science the optimal runtime for \np-complete problems is unknown. 
However,  the \emph{exponential-time hypothesis} \cite{lokshtanov2013lower} conjectures that for any algorithm attempting a solution at these problems there exist instances demanding an exponential runtime. 

Physically, one can invoke the Church-Turing-Deutsch principle \cite{deutsch1985quantum} that interprets computations as physical processes. 
For example, \np\, has been established to correspond to the cooling of spin glasses \cite{barahona1982computational}.
These materials are known to sometimes take an extremely long time to cool down. 
On the other hand, very many solid-state materials seem to cool down much faster. 
Indeed, insights in computer science suggest that the hardness of \np-complete problems lies in few tough instances with particularly rugged landscape. 
Phenomena like this are described in the framework of \emph{average-case} complexity. 
While \np-complete problems are unlikely to be hard on average \cite{feigenbaum_selfreducible_1991}, average-case hard problems are ubiquitous for the class \sharpP.
Recently, first examples directly relevant to demonstrating computational separation between classical and quantum devices have been pointed out \cite{aaronson_bosonsampling_2010,bouland_quantum_2018}. 

There are several approaches for a rigorous theory of average-case complexity. Arguably the most natural is \emph{random self-reducibility}: 
The idea is that a machine powerful enough to solve e.g.\ three quarters of the instances would allow to solve all instances. 
Thus, it becomes implausible to find heuristic algorithms that solve significant numbers of instances as the self-reducibility structure would imply efficiency even for those instances that are particularly hard. 
Hence, while \sharpP-hardness is a very strong statement, it does not preclude the existence of efficient practical algorithms that are capable of solving relevant instances. 

In this work, we provide strong complexity theoretical indications that the latter is not the case for generic PEPS due to a random self-reducibility structure that we uncover. 
This extends the worst-case \sharpP-hardness result  \cite{schuch_hardness_2007} to the average case and is an even more challenging obstruction to overcome.
Technically, we make an extensive use of the recent insightful work by \citet{bouland_quantum_2018}, where average-case hardness has been established in the context of quantum circuits, and we also employ some of the results established by \citet{aaronson_bosonsampling_2010}.

In certain special instances fast algorithms might still be feasible.
For example it is known that matrix-product states admit a polynomial time deterministic contraction algorithm \cite{PhysRevA.74.022320}.
But even in two dimensions, this can happen under strong physical assumptions forcing the problem to admit a local structure \cite{anshu2016local,schwarz_PEPS_2017}.
Additionally, for certain subclasses some heuristic algorithms 
 \cite{FinitePEPScontractions,levin2007tensor,PhysRevB.80.155131,PhysRevLett.103.160601,PhysRevB.81.174411,yang2017loop,PhysRevB.95.045117,
verstraete2004renormalization,PhysRevB.78.155117,vanderstraeten2016gradient,fishman2017faster,
doi:10.1143/JPSJ.64.3598,NISHINO199669,nishino1997corner,orus2009simulation,orus2012exploring,AugustineThermal,
doi:10.1137/050644756,PhysRevA.74.022320,arad2010quantum,anshu2016local,schwarz_PEPS_2017
}
(see Refs.\ \cite{ran2017review,FinitePEPScontractions} for reviews)
yield  results of practical importance \cite{corboz2010simulation,PhysRevB.84.041108,SciPostPhys.3.4.030,PhysRevLett.111.137204,PhysRevB.94.235146,yan2011spin,PhysRevB.91.224431,PhysRevB.93.060407,PhysRevLett.118.137202,chen2018non}. 
 Our average-case hardness result, however, suggests that these approaches could break down even for relevant PEPS instances as otherwise difficult computational problems would admit (quasi-) polynomial algorithms.

Physically, for disordered systems, one would expect any accurate ground state approximation by a PEPS to inherit the randomness of the Hamiltonian \cite{wahl2017efficient}. Hence in this setting, we provide evidence of intractability.
Oftentimes, however, further physical assumptions are justified: 
While these completely generic PEPS are relevant for the study of strongly disordered systems, 
in many practically meaningful settings (in particular in the study of topological order), the relevant PEPS are translation-invariant. 
Remarkably, a worst-to-average case reduction as described in this letter works just as well for translation-invariant systems but we are unaware of a hardness result in the worst case for such systems.

\paragraph*{Projected entangled pair states.}\label{section:PEPS}
Here we recall the definition of PEPS \cite{PEPSOld} and review the computational problem from Ref.~\cite{schuch_hardness_2007} concerning the contraction of PEPS. We consider a family of graphs $G=(V,E)$ with $|V|=N$. Every vertex $v$ stands for a local spin system described by a Hilbert space $\mathcal{H}_v:=\mathbb{C}^d$. The physical Hilbert space is, thus, $\mathcal H\coloneqq \mathcal{H}_v^{\otimes N}=(\mathbb{C}^d)^{\otimes N}$. In the projective construction of PEPS one thinks of every edge $e\in E$ as a maximally entangled state $\sum_{i=1}^D \ket i\ket i$ in a virtual $D$-dimensional spin systems. A specific PEPS is then described by linear operators $P^{[v]}:\mathbb{C}^{D}\otimes \cdots \otimes \mathbb{C}^D\to \mathbb{C}^d$. It is defined as the state vector in $\mathcal H$ resulting from the application of all $P^{[v]}$ for all $v\in V$. Note that by this the obtained PEPS is not necessarily normalized.
The virtual dimension is assumed to satisfy $D=\poly(N)$ and is called \emph{bond dimension}. In our discussion, it will be crucial to discriminate between the PEPS, which is a state vector in $\mathcal H$, and its specification $\left(P^{[v]}\right)_v$.  We will refer to the latter as \emph{PEPS-data}. A PEPS is called \emph{translation-invariant} if the local tensors satisfy $P^{[v]}=P^{[w]}=P$ for all $v,w\in V$. These states have already been proven to be immensely useful in condensed matter research but the full regime of applicability is still open. 
Here, we assume open boundary conditions but our results carry over to the periodic case too.
\paragraph*{PEPS evaluation.}
It strikes as a tremendous advantage that PEPS are described by polynomial data only. 
However, the physical problem we want to tackle remains notoriously difficult in that contraction of PEPS is computationally hard. This is needed for obtaining physical quantities of interest like expectation values of local observables.
Specifically, the following computational tasks are the essential ingredients of PEPS contraction algorithms:
\begin{problem}[PEPS-contraction]
\label{problem:peps contraction}
\begin{algorithmic}
	\REQUIRE
		 A graph $G$ and corresponding finite PEPS-data $\left(P^{[v]}\right)_v$ describing an unnormalized state $\ket \psi$ and with bond dimension $D = \poly(N)$.
	\ENSURE
		$\langle \psi \ket \psi$.
\end{algorithmic}
\end{problem}

It is one of the key insights in Ref.\ \cite{schuch_hardness_2007} that this problem is in fact \sharpP-complete for the case that $G$ is a square lattice.
In the following, we recall the arguments leading to this observation. The construction uses measurement based quantum computing \cite{raussendorf_oneway_2001,raussendorf_cluster_2003,gross2007novel}. Measurement based quantum computing based on cluster states
performs a computation by initializing the cluster state on a square lattice and successively applying local sharp (projective) measurements to the local qubits. This is a universal model of a quantum computer and we can use it to encode any quantum circuit in a PEPS with polynomially bounded bond dimension. Notice first that the cluster state is a PEPS with bond dimension $D=2$. However, the outcome of the quantum computation performed by the measurements depends on the random outcomes. This is dealt with by correcting the outcome with Pauli operators depending on the random outcomes. The PEPS encoding the quantum circuit is now obtained by applying an additional projector $\ketbra{a}{a}$, where $a$ is the outcome that does not give rise to a non-trivial Pauli-correction. Hardness follows from encoding the problem of counting solutions for a Boolean formula: Given a Boolean formula $f$, finding $\#_1(f):=|\{x,f(x)=1\}|$ is \sharpP-complete.

\paragraph*{Main result.}\label{sec:average}
 Let us consider a generic PEPS in the sense that all entries of the tensor $P^{[v]}$ are drawn independently at random from the finite precision approximation of the normal distribution centered around $0$ and with standard deviation $\sigma$.
 We will denote this Gaussian distribution with $\mathcal P\coloneqq\mathcal{N}_{\mathbb{C}}(0,\sigma)^{D^4dN}$. 
Our main result is the following theorem:

\begin{theorem}[Average case hardness of PEPS contraction]\label{theorem:main}
	Suppose there exists an algorithm $\mathcal{O}$ that solves Problem~\ref{problem:peps contraction} for square lattices in polynomial time with probability $\frac34+\frac{1}{\poly(N)}$ when instances are drawn from $\mathcal{P}$. Then, there exists a randomized algorithm $\mathcal{O}'$ that solves any instance of Problem \ref{problem:peps contraction} in polynomial time with exponentially high probability $1-2^{-\poly(N)}$.
\end{theorem}
This rigorous statement can be interpreted in several intuitive ways. Firstly, it rules out the possibility that the computational hardness could be hidden in particular instances that are intractable, as it says that one could use the algorithm $\mathcal{O}$ to construct an algorithm $\mathcal{O}'$ that is efficient for all inputs.
 Colloquially, assuming that most instances are easy with a known heuristic $\mathcal{O}$, then the full problem will be equally easy. 
Secondly, it is important to note that Problem \ref{problem:peps contraction} requires exact computation \cite{schuch_hardness_2007} but a different variant of Theorem \ref{theorem:main} that we prove in the appendix shows the following: Exponential precision approximation is also intractable on average, however, under stronger requirements on the algorithm $\mathcal{O}$. 
Hence, structurally, we see that if \sharpP-problems are non-trivial, then it cannot be due to very rare instances. 
Our choice of the probability distribution is similar to that of Ref.\;\cite[Section 9.1]{aaronson_bosonsampling_2010}, where the evaluation of the so-called \emph{permanent} is considered which is also a \sharpP-complete computational problem.
 Therefore, Theorem \ref{theorem:main} shows not only that both these problems are in the same complexity class, but they also have the same complexity theoretical structure.
 Note that the result holds for arbitrary graphs as well, though the statement is trivial in one dimension \cite{PhysRevA.74.022320}.
\paragraph*{Random self-reducibility.}
There are several precise mathematical candidates for a definition of \emph{average-case hardness}. We find that PEPS-contraction is average-case hard in the same sense as canonical combinatorial problems \cite{aaronson_bosonsampling_2010, lipton_permanent_1991}: Both problems admit random self-reducibility. A problem is \emph{randomly self-reducible} if the evaluation of any instance $x$ can be reduced to the evaluation of random instances $y_1,\dots,y_k$ with a bounded probability independent of the input. We will sketch how this is done for the permanent and PEPS giving the essential proof idea, see Ref.~\cite{bouland_quantum_2018} for a particularly clear exposition in the context of quantum circuits.
The complete argument can be found in Appendix \ref{appendix:proof}.

In a seminal result, \citet{lipton_permanent_1991} proved random self-reducibility for the evaluation of the permanent, a function that takes as an input a square matrix and outputs a number. The permanent of a matrix $A\in \mathbb{C}^{n\times n}$ is defined as the 'determinant without signs':
\begin{equation}
\mathrm{perm}(A)\coloneqq\sum_{\sigma\in S_n}\prod_{i=1}^n A_{i,\sigma(i)} \, , 
\end{equation}
where $S_n$ is the symmetric group. 
 However, very unlike the determinant, the permanent turns out to yield a difficult combinatorial problem: Its evaluation has been proven to be \sharpP-complete by \citet{valiant_permanent_1979}. The proof of random-self reducibility is rooted in the algebraic fact that the permanent defines a polynomial of degree $n$ in the entries of its input matrix $A$.
More precisely, the strategy is to take any (hard) instance $A$ that we want to compute, draw a uniformly random matrix $B$ and define 
\begin{align}
\label{eq:liptonrandome}
E(t)\coloneqq A+tB \, , 
\end{align}
for $t \in \mb R$. 
Notice that $E(t)$ is uniformly random for any $t$ because $B$ is, even though $E(t)$ and $E(t')$ are correlated.
The permanent of these matrices is a polynomial $q(t) \coloneqq \mathrm{perm}(E(t))$ of degree\;$n$. Even if the algorithm $\mathcal{O}$ fails to accurately output $\mathrm{perm}(A)$ it will, by assumption, likely correctly evaluate $q(t_i)$ for a choice of $t_i$. The idea is to infer $q(0)$ from the values at $\{t_i\}$ via polynomial interpolation. We will explain this step in more detail in the next paragraph for the setting of PEPS.

\paragraph*{Sketch of proof for Theorem \ref{theorem:main}.}
Let us sketch how the worst to average-case reduction works for PEPS contractions. For a detailed and formal proof we refer to Appendix \ref{appendix:proof}.
First, notice that given a bond dimension $D$, the set of possible PEPS-data admits a canonical vector space structure defined by
\begin{align*}
\left(P_1^{[v]}\right)_v+\left(P_2^{[v]}\right)_v&:=\left(P_1^{[v]}+P_2^{[v]}\right)_v\\
\lambda\left(P_1^{[v]}\right)_v&:=\left(\lambda P_1^{[v]}\right)_v, \;\;\;\;\;\lambda\in\mathbb{C}.
\end{align*}
Notice that already in this step, it is crucial to discriminate between PEPS-data and PEPS since the above definition has very little to do with the addition of the corresponding states. Intuitively, we scramble independently the individual tensors.
Given a hard instance $(P^{[v]})_v$, we draw random PEPS-data $(Q^{[v]})_v$ and define
\begin{equation}
\left(R(t)^{[v]}\right)_v \coloneqq t\left (P^{[v]}\right)_v+(1-t)\left(Q^{[v]}\right)_v.
\end{equation}
Thus, $(R(0)^{[v]})_v=(Q^{[v]})_v$ and $(R(1)^{[v]})_v=(P^{[v]})_v$.
This choice of a scrambled operator is suitable for us because it allows us to deal with a subtlety arising from the fact that the PEPS-data $(R(t)^{[v]})_v$ is not Gauss-random even though $(Q^{[v]})_v$ is.
This is different to the setting of \citet{lipton_permanent_1991} but has been worked out for boson sampling \cite{aaronson_bosonsampling_2010}, where it was shown that the difference is immaterial for small $t$.
This carries over to our case as we discuss in Appendix \ref{appendix:proof}. 
We choose $k=\poly(N)$ sampling points $t_i\in [0,\varepsilon)$, where $\varepsilon$ is polynomially small. With these sampling points we perform polynomial interpolation.

Let $\ket{\psi(t_i)}$ denote the PEPS corresponding to the data $(R(t_i)^{[v]})_v$. 
In analogy to the discussion of the permanent, we define the function $q(t)\coloneqq\langle \psi(t)\ket{\psi(t)}$, which is a polynomial in $t$ of degree $r=2N$. For each sampling point, the machine $\mathcal O$ performs the exact contraction with probability $\frac34+\frac{1}{\poly(N)}$. Using the Markov inequality, we obtain that out of the $k$ sampling points $\mathcal{O}$ outputs the correct value of the contraction $q$ for at least $\frac{k+r}{2}$ with probability $\frac12+\frac{1}{\poly(N)}$. Provided $k>r$, we can use polynomial interpolation to reconstruct the coefficients of the polynomial $q$ such that $q(1)$ is the desired PEPS contraction value. This is achieved by the so-called \emph{Berlekamp-Welch algorithm}, a result in computer science, which in polynomial runtime outputs the coefficients of $q$. Thus, using a small computational overhead, we obtain $q(1)$. Repeating this procedure, and taking the majority vote choosing the most frequent outcome, the probability of success can be amplified to $1-2^{-\poly(N)}$. We define this final outcome to be the output of the algorithm $\mathcal O'$.
\paragraph*{Translation invariance.}
In many physical applications, e.g. in solid state materials or systems admitting topological order, the system of interest is translation-invariant.
Hence, the PEPS-data should reflect this symmetry and one would naturally set all local tensors to be equal. 
In this case we do not know the corresponding computational problem to be \sharpP-hard, for example the \sharpP-hard instances in
Ref.~\cite{schuch_hardness_2007} are not translation-invariant. 
However, our worst-to-average case reduction works just as well in this special case, simply by choosing $(Q^{[v]})_v=(Q)_v$, where $Q$ is drawn from $\mathcal{N}_{\mathbb{C}} (0,\sigma)^{D^4d}$. The same argument and statement of the main theorem goes through.
 This leaves us with two mutually exclusive options: If the translation-invariant problem is hard for a complexity class $C$, then it follows that the problem is $C$-hard on average in the sense of our main theorem. If the problem is merely in \p, then it is enough to find a heuristic for about $\frac34$ of the inputs to find a full randomized algorithm.
 On the other hand, if $C=\sharpP$, then even the translation-invariant PEPS contraction problem would appear to be average-case intractable. We are unaware of random self-reducibility results for complexity classes other than \sharpP. We thus expect a dichotomy: Either the translation-invariant problem is in \p \, or it is \sharpP-complete.
\paragraph*{Evaluation precision.}
As far as we know, it is state of the art in computer science to prove random self-reducibility structures for problems given the promise that $\mathcal O$ works with 
at least exponential precision.
In fact, we can improve our main theorem for this case too, at the cost of requiring $\mathcal O$ to function with a probability of $1-\frac{1}{12N}$. 
The reason for this trade-off is that subtleties arise in the technical steps, where the Berlekamp-Welch algorithm has to be replaced with a noise-resistant method.
However, in the bigger picture, it does not seem possible to extend the seminal idea of Lipton to $\mathcal O$ working with polynomial precision.
Intuitively, we interpolate around small $t_i$ and want to evaluate at $1$. We consider it unlikely that it is possible to devise an extrapolation method which accurately outputs $q(1)$.
However, if it turned out to be the case, e.g., by future results in computer science, then our worst-to-average case reduction would work for the case of practical interest, i.e., polynomial precision of $\mathcal O$.
Related questions of precision relaxation are of interest in quantum information theory in the context of searching for quantum speed-ups.
Here, certain precision relaxations are conjectured to be average case hard as well \cite{aaronson_bosonsampling_2010,bouland_quantum_2018}.

\paragraph*{Expectation values.}
 The computational Problem \ref{problem:peps contraction} is concerned with PEPS contractions. The quantity that one computes is the norm of the respective PEPS. However, in most physical applications the quantities of interest are expectation values of a local observable $\hat{A}$
 \begin{equation}
  \langle \hat{A}\rangle_{\psi}=\frac{\langle\psi|\hat{A}|\psi\rangle}{\braket{\psi}{\psi}}.
 \end{equation}
   Notice that this problem and its unnormalized version have both been proven to be \sharpP-complete in Ref.\ \cite{schuch_hardness_2007} as well. For any algorithm that uses PEPS normalization as an intermediate step our main theorem is directly of interest and reflects the fundamental structure of the problem at hand. 
  In the general case we can prove a worst-to-average result for this quantity as well. It is easy to see that our discussion of PEPS contraction carries over to the discussion of unnormalized expectation values. We show that a close analogue of Theorem \ref{theorem:main} holds for this quantity as well. The normalized expectation value is slightly more subtle in the following sense: The analogue of the function $q$ is not a polynomial but a rational function $\frac{q}{p}$ where the degrees of both polynomials $q$ and $p$ are bounded by $2N$.
We can simply solve for the coefficients on enough sampling points to obtain the respective coefficients.
This, however, requires a stronger machine $\mathcal{O}$. This result might be further improved by the use of more sophisticated algorithms for the reconstruction of rational functions. 
   
\paragraph*{Implications on practical tensor network algorithms.} The results found here have interesting implications to the performance of 
PEPS contraction algorithms aimed at solving condensed-matter problems \cite{orus2014practical,verstraete2008matrix,schollwock2011density}. 
There are three insights that are important in this respect: Firstly, the results laid out here relate average-case to worst-case complexity. 
In that, they apply to any tensor network contraction algorithm as the structure of random self-reducibility shows that  if a given method $\mathcal O$ has trouble at less than a quarter of instances, these can be in principle also treated with a small polynomial runtime overhead by our construction of the randomized algorithm $\mathcal O'$ 
(and, for that matter, our results also pertain to algorithms in \p).
Secondly,  it is known that PEPS contraction algorithms often work well in practice for reasonable condensed-matter systems \cite{FinitePEPScontractions,ran2017review} which may seem at first sight at odds with the results presented here and in Ref.\ \cite{schuch_hardness_2007}.
For this, one has to acknowledge that many important problems have additional structure that may render the PEPS contraction feasible.
Specifically, it was proven in Ref.~\cite{schwarz_PEPS_2017} that local normalized expectation values of injective PEPS with \emph{uniformly 
gapped parent Hamiltonian} can be evaluated in quasi-polynomial time, i.e., faster than conjectured by the exponential-time hypothesis. 
Following up on this observation, it seems conceivable that one can devise PEPS algorithms that provide ground states of systems in a trivial phase (possibly even with convergence proofs), by making use of techniques of quasi-adiabatic evolution \cite{NachtergaelePhases,HastingsQuasiAdiabatic}, applying short circuits to product states as ground states of trivial parents.
Having said that, any such approach would require keeping track of ground states of families of Hamiltonians.
Thirdly, in most practical algorithms used in practice, in contrast, some initial condition for the PEPS is chosen, which is iteratively refined via sweeps, until a good convergence to the ground state is encountered. 
In fact, in practice, the PEPS data are initially often chosen randomly, following a refinement in sweeps by iteratively minimizing the energy evaluated from a local Hamiltonian. 
The results laid out here show that it is crucial to devise meaningful schemes making reasonable choices of these initial conditions. But our average-case
 hardness results of PEPS contraction indicate that one should be particularly cautious when choosing such initial states. 
 
\paragraph*{Outlook.}
In this work we presented the first average case complexity result in the context of quantum many-body systems, specifically tensor network states.
Our main result is structural, namely we prove that the hard instances of PEPS-contraction make up a significant fraction of all instances.
Physically, this means that contraction of PEPS with random tensors is likely to be computationally hard to accurately evaluate. 
Conceptually, we establish structural similarities to the evaluation of the permanent.
Our results hold under the assumption of accurate or exponential precision.
In Appendix C, we stress that also on physical grounds, to demand exponential precision is very much reasonable.
However, in a physical context it is often sufficient to evaluate observables up to polynomial precision.
 The major \emph{open problem} is thus to extend the presented analysis to this case.
For PEPS contractions establishing such a result would have direct practical implications.
Furthermore, we are not aware of any \sharpP-completeness result for translation-invariant PEPS.
Thus, the general \emph{open question} should be: What are the instances of PEPS for which known contraction methods have convergence guarantees?
It is our hope that further research at the interface between computer science and quantum many-body physics will provide exciting insights to this question.
\begin{acknowledgments}
	We thank Adam Bouland, Bill Fefferman, Juan Bermejo-Vega, Christopher Chubb, Ingo Roth, Alexander Nietner, Augustine Kshetrimayum and Radu Curticapean  for valuable discussions and comments.
  In particular we are grateful to Adam Bouland and Bill Fefferman for pointing out to us the result of \citet{rakhmanov_interpolation_2007}. This work was supported
  by the DFG (EI 519/14-1, EI 519/7-1, CRC 183 B1), the ERC (TAQ), and the Templeton Foundation.
\end{acknowledgments}

\bibliographystyle{apsrev4-1}

%

\appendix
\section{Variants and proof of Theorem \ref{theorem:main}}\label{appendix:proof}
As explained in the main text, our result comes in different flavors. Here, we present our results in full technical detail. 
We formalize the problem of evaluating expectation values of local observables in the following two problems:
\begin{problem}[PEPS-contraction:UEV]
	\label{problem:peps contractionUEV}
	\begin{algorithmic}
		\REQUIRE
		The same input as in Problem \ref{problem:peps contraction} and additionally a local observable $\hat{A}$.
		\ENSURE
		$\langle \psi |\hat{A}|\psi\rangle$.
	\end{algorithmic}
\end{problem}

\begin{problem}[PEPS-contraction:NEV]
	\label{problem:peps contractionNEV}
	\begin{algorithmic}
		\REQUIRE
		The same input as in Problem \ref{problem:peps contraction} and additionally a local observable $\hat{A}$.
		\ENSURE
		$\langle\psi|\hat{A}|\psi\rangle/{\braket{\psi}{\psi}}$.
	\end{algorithmic}
\end{problem} 

We prove all results for two canonical choices: The first is to draw entry-wise from a uniform distribution centered around zero and truncated at some chosen threshold $\sigma$, which we will denote by $\mathcal{U}=\mathcal{U}_{\mathbb{C}}(0,\sigma)$ and the product distribution by $\mathcal{P}_1:=\mathcal{U}^{D^4dN}$. Almost equivalently we could draw from a Gaussian distribution. We will denote this Gaussian distribution in this appendix with $\mathcal{P}_2:=\mathcal{G}^{D^4dN}:=\mathcal{N}_{\mathbb{C}}(0,\sigma)^{D^4dN}$. This is reminiscent to a discussion about the permanent with entries in the complex numbers in Ref.\ \cite[Section 9.1]{aaronson_bosonsampling_2010}. 
More precisely, we prove the following technical theorems:

\begin{theorem}[Worst-to-average reduction]\label{theorem:precision12}
	Suppose there exists a machine $\mathcal{O}$ that solves Problem~\ref{problem:peps contraction} or Problem \ref{problem:peps contractionUEV} within precision $2^{-\poly N}$ for square lattices in polynomial time with a probability of $1-\frac{1}{12N}$ over the instance drawn from $\mathcal{P}_i$ for $i=1,2$. Then, there exists a machine $\mathcal{O}'$ that solves any instance with precision $2^{-\poly(N)}$ of the respective problem in randomized polynomial time with exponentially high probability.
\end{theorem}
We will prove this theorem first, as it requires the most technical work.
If we do not relax to exponential precision but require perfect arithmetical evaluation of the machine $\mathcal O$, we obtain a much stronger worst-to-average reduction:

\begin{theorem}[Stronger worst-to-average reduction]\label{theorem:main2}
	Suppose it exists a machine $\mathcal{O}$ that solves Problem \ref{problem:peps contraction} or \ref{problem:peps contractionUEV} exactly for square lattices in polynomial time with a probability of $\frac{3}{4}+\frac{1}{\poly N}$ drawn from $\mathcal{P}_i$, with $i=1,2$. Then, there exists a machine that solves any instance of the respective problem in randomized polynomial time with exponentially high precision.
\end{theorem}
Notice that Theorem \ref{theorem:main} is a special case of the above. Namely, it corresponds to the choice of Problem \ref{problem:peps contraction} and probability distribution $\mathcal P=\mathcal{P}_2$.
Finally, again requiring perfect evaluation, we obtain worst-to-average reduction for the normalized expectation value problem as well:

\begin{theorem}[Normalized expectation values]\label{theorem:main_NEV}
	Suppose it exists a machine $\mathcal{O}$ that solves Problem \ref{problem:peps contractionNEV} exactly for square lattices in polynomial time with a probability of $1-\frac{1}{24N}$ drawn from $P_i$ with $i=1,2$. Then there exists a machine that solves any instance of the respective problem in randomized polynomial time with exponentially high precision.
\end{theorem}

\subsection{Proof of Theorem \ref{theorem:precision12}}
Before we turn to presenting the proof, we state a lemma which resembles Lemma 48 in Ref.\ \cite{aaronson_bosonsampling_2010}.	Let us denote with $\mathcal{N}_{\mathbb{C}}(\mu,\sigma)$ the normal distribution over the complex numbers with mean $\mu$ and standard deviation $\sigma$. The lemma establishes that products of normal distributions with small mean are close to a product of the standard normal distribution with zero mean.
\begin{lemma}[Gaussian distributions]\label{lemma:gaussiandistribution}
 For the distributions 
	\begin{align}
	\mathcal{D}_1&:=\mathcal{N}_{\mathbb{R}}\left(0,(1-\varepsilon)^2\sigma\right)^M,\\
	\mathcal{D}_2&:=\prod_{i=1}^M\mathcal{N}_{\mathbb{R}}(v_i,\sigma)
	\end{align}
	with $v\in \mathbb{C}^M$, it holds that
	\begin{align}
	||\mathcal{D}_1-\mathcal{N}_{\mathbb{R}}(0,\sigma)^M||&\leq 2M\varepsilon,\\
	||\mathcal{D}_2-\mathcal{N}_{\mathbb{R}}(0,\sigma)^M||&\leq\frac{1}{\sigma}||v||_1,
	\end{align}
	where $||\bullet||$ denotes the total variation distance and $v\in \mathbb{C}^M$. The same result holds if we substitute $\mathcal{N}$ with $\mathcal{U}$.
\end{lemma}
\begin{proof}[Proof of Lemma \ref{lemma:gaussiandistribution}]
	We prove the lemma for the Gaussian case. The uniform can be obtained similarly. We obtain with the triangle inequality for the total variation distance:
	 \begin{equation}
     ||\mathcal{D}_1-\mathcal{G}^M||\leq M||\mathcal{N}_{\mathbb{R}}(0,(1-\varepsilon)^2\sigma)-\mathcal{N}_{\mathbb{R}}(0,\sigma)||.
     \end{equation}
     With the relation between total variation distance and $L^1$-norm, we obtain
     \begin{align}
      &||\mathcal{D}_1-\mathcal{G}^M||\nonumber\\
      &\leq \frac{M}{2}\int_{-\infty}^{\infty}\left|\frac{1}{\sqrt{2\pi\sigma}}e^{-\frac{x^2}{2\sigma^2}}-\frac{1}{\sqrt{2\pi}\sigma(1-\varepsilon)}e^{-\frac{x^2}{2\sigma^2(1-\varepsilon)^2}}\right|\mathrm{d}x\nonumber\\
      &=\frac{M}{2\sqrt{2\pi}\sigma (1-\varepsilon)}\int_{-\infty}^{\infty}\left|(1-\varepsilon)e^{-\frac{x^2}{2\sigma^2}}-e^{-\frac{x^2}{2\sigma^2(1-\varepsilon)^2}}\right|\mathrm{d}x\nonumber\\
      &\leq \frac{M\varepsilon}{2\sqrt{2\pi}\sigma(1-\varepsilon)}\int_{-\infty}^{\infty}e^{-\frac{x^2}{2\sigma^2}}\nonumber\\&+\frac{M}{2\sqrt{2\pi}\sigma(1-\varepsilon)}\int_{-\infty}^{\infty}e^{-\frac{x^2}{2\sigma^2}}-e^{-\frac{x^2}{2\sigma^2(1-\varepsilon)^2}}\mathrm{d}x\nonumber\\
      &=\frac{M\varepsilon}{2(1-\varepsilon)}+\frac{M}{2(1-\varepsilon)}-\frac{M}{2}=\frac{M\varepsilon}{1-\varepsilon}\leq 2M\varepsilon.
     \end{align}
     The second inequality follows using again the triangle inequality:
     \begin{align}
     &||\mathcal{D}_2-\mathcal{G}^M||\leq\sum_{i=1}^M\left|\left|\mathcal{N}_{\mathbb{R}}(v_i,\sigma)-\mathcal{N}_{\mathbb{R}}(0,\sigma)\right|\right|\nonumber\\
     &=\sum_{i=1}^M\frac{1}{2}\int_{-\infty}^{\infty}\left|\frac{1}{\sqrt{2\pi}\sigma}e^{-\frac{(x-v_i)^2}{2\sigma^2}}-\frac{1}{\sqrt{2\pi}\sigma}e^{-\frac{x^2}{2\sigma^2}}\right|\mathrm{d}x\nonumber\\
     &=\sum_{i=1}^M\frac{1}{2\sqrt{2\pi}}\int_{-\infty}^{\infty}\left|e^{-\frac{(x-v_i/\sigma)^2}{2}}-e^{-\frac{x^2}{2}}\right|\mathrm{d}x\nonumber\\
     &\leq \sum_{i=1}^M\frac{|v_i|}{\sigma}=\frac{||v||_1}{\sigma},
     \end{align}
     where the last inequality follows from a straightforward calculation.
\end{proof}
\begin{proof}[Proof of Theorem \ref{theorem:precision12}] 
	For simplicity, we set $\sigma=1$. Furthermore, we restrict to the case of Problem \ref{problem:peps contraction} as the proof for the case of Problem \ref{problem:peps contractionUEV} is completely analogous.
	Consider Problem \ref{problem:peps contraction} and a hard instance defined by the data $(P^{[v]})_v$, e.g. the encoding of a Boolean function as was done in
	Ref.\ \cite{schuch_hardness_2007}. It suffices to consider a $(P^{[v]})_v$ with all matrix entries being bounded by $1$ as all instances constructed in 
	Ref.\ \cite{schuch_hardness_2007} admit this form. Furthermore, we draw PEPS-data from the standard Gaussian distribution entry-wise, denoted as $\left(Q^{[v]}\right)_v\sim \mathcal{G}^{D^4dN}$. Analogously to \citet{lipton_permanent_1991}, we define 
	\begin{equation}
	\left(R(t)^{[v]}\right)_v := t\left (P^{[v]}\right)_v+(1-t)\left(Q^{[v]}\right)_v.
	\end{equation}
	Now, let $\ket{\psi(t)}$ denote the PEPS corresponding to this data.
	In analogy to the discussion of the permanent, we define the function $q(t):=\langle \psi(t)\ket{\psi(t)}$. 
	Notice that this function is a polynomial in $t$ with degree $r=2N$, which scales polynomially in the input length. 
	Before we can apply Theorem \ref{theorem:BerlekampWelch}, we have to deal with the fact that the $(R(t)^{[v]})_v$ are not distributed according to the Gaussian distribution. We will need only very small $t$ bounded by some $\varepsilon>0$, such that the difference between the respective distributions is immaterial. Specifically, the $(R(t)^{[v]})_v$ tensors are distributed according to
	\begin{equation}
	\mathcal{D}=\prod_{i=1}^{D^4dN}\mathcal{N}_{\mathbb{C}}\left(tp_i,(1-t)^2\right).
	\end{equation}
	Thus, from a triangle inequality and Lemma \ref{lemma:gaussiandistribution}, we obtain
	\begin{equation}\label{eq:totalvariation}
	\left|\left|\mathcal{D}-\mathcal{G}^{D^4dN}\right|\right|\leq \left(4D^4dN+2D^4dN\right)\varepsilon= \left(6D^4dN\right)\varepsilon
	\end{equation}
	for $|t|\leq \varepsilon$, by identifying $\mathbb{C}$ with $\mathbb{R}^2$. It will suffice to set 
	\begin{equation}
	\varepsilon:=\frac{\delta}{6D^4dN}
	\end{equation}
	and $\delta := \frac{1}{12N}$.
	This implies that for a small enough inverse polynomial $\varepsilon$, we can make the total variation distance polynomially small.
	Let $\{t_i\}_{i\in [r+1]}$ be the set of $r+1$ equidistant points in $[0,\varepsilon]$. We will now use the assumption from the theorem's statement that the machine $\mathcal{O}$ works for a $1-\delta$ fraction of the instances drawn from $\mathcal{G}^{D^4dN}$. Using \eqref{eq:totalvariation}, we obtain for the success probability of the machine evaluating at the points $t_i$ accurately up to within precision $2^{-\poly N}$
	\begin{align}
	&\Pr\left[\left|\mathcal{O}\left(\left(R^{[v]}\right)_v(t_i)\right)-q(t_i)\right|\leq 2^{-\poly N}\right]\nonumber\\
	&\geq 1-\delta-\left|\left|\mathcal{D}-\mathcal{G}^{D^4dN}\right|\right|\nonumber\\
	&\geq 1-2\delta,
	\end{align}
	where we used that the total variation distance is an upper bound on the difference in probability the two distributions could possibly assign to an event.
	
	 Finally, we obtain the probability of $r+1$ consecutive succesful evaluations as
	\begin{align}
	&\Pr\left[\left|\left\{i\in [r+1],|\mathcal{O}(t_i)-q(t_i)|\leq 2^{-\poly N}\right\}\right|= r+1\right]\nonumber\\
	&=\left(1-2\delta\right)^{r+1}=\left(1-\frac{1}{6N}\right)^{r+1}\nonumber\\
	&\geq 1-\frac{2N+1}{6N}= \frac{2}{3}-\frac{1}{6N},
	\end{align}
 by Bernoulli's inequality. Here, we abbreviated $\mathcal{O}\left(\left(R^{[v]}\right)_v(t_i)\right)$ with $\mathcal{O}(t_i)$. Given the evaluation values at the $t_i$, we can solve for the coefficients and obtain a polynomial $\tilde{q}$ which satisfies $|\tilde{q}(t_i)-q(t_i)|\leq 2^{-\poly N}$ for all $t_i$ with high probability. The machine $\mathcal{O}'$ then evaluates $\tilde{q}(1)$, which is an estimate for $q(1)=\braket{\psi}{\psi}$.
 
  To bound the error on this estimate we will use two powerful results: The first on noisy extrapolations and the second on noisy interpolations of polynomials. A version of the following lemma was proven in Ref.\ \cite{paturi_extrapolation_1992}, see also Ref.\ \cite{aaronson_bosonsampling_2010}[Section 9.1].
	
	\begin{lemma}[Paturi]\label{lemma:Paturi}
		Let $p:\mathbb{R}\to\mathbb{R}$ be a polynomial of degree $r$ and suppose $|p(x)|\leq\Delta$ for all $x$ such that $|x|\leq\varepsilon$. Then, $|p(1)|\leq \Delta e^{2r(1+1/\varepsilon)}$.
	\end{lemma}
	The following theorem was proven in \citet{rakhmanov_interpolation_2007}.
	\begin{theorem}[Rakhmanov]\label{theorem:Rakmanov}
		Let $E_k$ denote the set of $k$ equidistant points in $(-1,1)$. Then, for a polynomial $p:\mathbb{R}\to\mathbb{R}$ with degree $r$ such that $|p(y)|\leq 1$ for all $y\in E_k$, it holds that 
		\begin{equation}
		|p(x)|\leq C\log\left(\frac{\pi}{\arctan\left(\frac{k}{r}\sqrt{\mathcal{R}^2-x^2}\right)}\right)
		\end{equation}
		with
		\begin{equation}
		|x|\leq \mathcal{R}:=\sqrt{1-\frac{r^2}{k^2}}.
		\end{equation}
	\end{theorem}
	We will use the second result to bound the error between the points and then use the first result to bound the error on $\tilde{q}(1)$. For the proof, we shift the polynomial $p$ such that the intervall of interest is centered around the origin. Furthermore, we can straightforwardly implement that we work with a smaller interval. We obtain that
	\begin{equation}
	\mathcal{R}=\sqrt{1-\frac{r^2}{(r+1)^2}}\frac{\varepsilon}{2}=\sqrt{\frac{4N+1}{(2N+1)^2}}\frac{\varepsilon}{2}.
	\end{equation}
	Restricting to the strict subinterval $[-\frac{\mathcal{R}}{2},\frac{\mathcal{R}}{2}]$, we can apply Theorem \ref{theorem:Rakmanov} and obtain the following bound for all $t\in [-\frac{\mathcal R}{2},\frac{\mathcal R}{2}]$,
	\begin{align}
	|p(t)|&\leq 2^{-\poly N}C\log\left(\frac{\pi}{\arctan\left(\frac{k}{r}\sqrt{\mathcal{R}^2-x^2}\right)}\right)\nonumber\\
	&\leq 2^{-\poly N} C \log\left(\frac{\pi}{\arctan(2\mathcal R)}\right)\leq 2^{-\frac{1}{2}\poly N}.
	\end{align}
	Finally, we can apply Lemma \ref{lemma:Paturi}. This yields the desired bound on the difference between the estimate $\tilde{q}(1)$ and the actual value $q(1)$:
	\begin{align}
	|\tilde{q}(1)-q(1)|=&|p(1)|\leq 2^{-\frac{1}{2}\poly N+4\log_2(e)N(1+2/\mathcal{R})}\nonumber\\
	=&2^{-\poly'N}
	\end{align}
	for a sufficiently large $\poly$.
	Finally, we remark that the success probability can be exponentially amplified by repeating the above procedure polynomially many times because of the Chernoff bound.
	
\end{proof}
\subsection{Proof of Theorem \ref{theorem:main2}}
The superior bound in Theorem \ref{theorem:main2} follows from the fact that we can invoke the Berlekamp-Welch algorithm in the interpolation step.
The latter is a provably correct algorithm for the interpolation of polynomials due to Ref.~\cite{welch_polynomials_86}. Compare also \citet{bouland_quantum_2018}.

\begin{theorem}[Berlekamp-Welch \cite{welch_polynomials_86}]\label{theorem:BerlekampWelch}
	Let $q$ be a degree-$r$ polynomial over any field $\mathbb{F}$. 
	Suppose we are given $k$ pairs of elements $\{(x_i,y_i)\}_{i=1,\dots,k}$ with all $x_i$ distinct with the promise that $y_i=q(x_i)$ for at least $\max(r+1, (k+r)/2)$ points. 
	Then, one can recover $q$ exactly in $\poly(k,r)$ deterministic time.
\end{theorem}

 Analogous to the proof of Theorem \ref{theorem:main}, we arrive at a polynomial $q(t)=\braket{\psi(t)}{\psi(t)}$ of degree $r=2N$. Instead of $r+1$ queries to the machine $\mathcal{O}$, we query it $k=\poly(N)$ times. Berlekamp-Welch requires that at least $\frac{k+r}{2}$ of obtained $k$ data points are correct in order to reconstruct the polynomial. We furthermore assume that $k>r$. From Markov's inequality we obtain:
\begin{align}
&\Pr\left[\left|\left\{i,\mathcal{O}(t_i)=q(t_i)\right\}\right|\geq \frac{k+r}{2}\right]\geq 1-\frac{2\mathbb{E}}{k-r}\nonumber\\
&\geq 1-\frac{2(\frac{1}{4}-\frac{1}{\poly N})k}{k-r}= 1-\frac{k}{2(k-r)}+\frac{2k}{\poly(N)(k-r)}\nonumber\\
&=\frac{1}{2}-\frac{r}{2(k-r)}+\frac{2k}{\poly(N)(k-r)},
\end{align}
where we abbreviate the expectation value in question with $\mathbb{E}$.
Thus, by choosing $k$ polynomially large, we obtain an expression that is polynomially close to $\frac12$.
Again, by repeating the procedure a polynomial number of times and taking a majority vote we can amplify this probability exponentially. With this probability, the Berlekamp-Welch algorithm outputs $q$ exactly and we can simply evaluate $q(1)$ without having to worry about the error of extrapolation. It seems appropriate to point out that we are in fact not drawing data from the Gaussian distribution in this case but from a discrete analogue of it. However, this does not change the details of our analysis.

\subsection{Proof of Theorem \ref{theorem:main_NEV}}
Here, we need the exact evaluation because we are not aware of results analogous to Lemma \ref{lemma:Paturi} and Theorem \ref{theorem:Rakmanov} for rational functions. Also, there does not seem to be an analogue of the Berlekamp-Welch algorithm. Nevertheless, we know that the function we are interested in can be described by the quotient of two polynomials of degree at most $r=2N$. This leaves us with $4N+1$ unknown coefficients. By invoking the machine $\mathcal{O}$ a total of $2r+1=4N+1$ times, we obtain that all queries are correct with a probability of 
\begin{equation}
\Pr=(1-2\delta)^{2r+1}\geq 1-\frac{4N+1}{12N}\geq\frac{2}{3}-\frac{1}{12N}.
\end{equation}
The remainder of the proof follows analogously to the proof of Theorem \ref{theorem:main2}.
\section{Liptons theorem}
The original theorem proven by Lipton is formulated for the permanent of matrices in \emph{finite fields}:

\begin{theorem}[\citet{lipton_permanent_1991}]
	\label{thm:lipton}
	Let $\mathrm{perm}:\mathbb{F}_{q}^{n\times n}\to\mathbb{F}_q$ be the permanent defined via
	\begin{equation}
	\mathrm{perm}(A):=\sum_{\sigma\in S_n}\prod_{i=1}^n A_{i,\sigma(i)} \, , 
	\end{equation}
	where by $\mathbb{F}_q$ we denote the set of $n \times n $ matrices over the finite field $\mathbb{F}_q$ with (prime) characteristic $q$ and $S_n$ the $n$-th symmetric group.
	Evaluating this quantity admits random self-reducibility in the following sense: 
	for sufficiently (polynomially) large $q$, the capacity to evaluate the permanent with probability $\geq \frac{3}{4}+\frac{1}{\poly (n)}$ for a uniformly random matrix $M \in \mathbb{F}_q^{n\times n}$, implies the capacity to determine the permanent of any given matrix $A$ with probability $1- \delta$ for an exponentially small $\delta$.
\end{theorem}
A variant of this theorem for the field $\mathbb{C}$ was proven in \cite[Section 9.1]{aaronson_bosonsampling_2010}.
\section{Exponential dependence on PEPS data}\label{app:exponential}

The argument in the main text emphasizes the demanding precision that is required when specifying the PEPS data.
In this section, we stress that this is not merely done for complexity-theoretic reasons: A pair of states can be defined by very similar
PEPS data, while their norms can be vastly different. In fact, to specify the norm of a PEPS, one needs exponential precision in the PEPS data,
as a moment of thought reveals. This is already true in one spatial dimension for matrix product states.  Take $D=2$, $d=2$, an a
translation-invariant open boundary
condition MPS, so that the vertex set $V$ is that of $N$ sites, $E$ reflecting nearest neighbor interactions. The linear operators
$P^{[v]}=P$ are for all $v$
 defined by
\begin{equation}
	P^{[v]}= \sum_{i=1,2} \sum_{\alpha, \beta=1}^D
	A[i]_{\alpha,\beta}|i\rangle\langle \alpha,\beta|,
\end{equation}	
where for the state vector $\ket{\psi}$ we take
\begin{eqnarray}
	A[0]:=\text{diag}(1,0),\,
	A[1]:=
	\text{diag}(0,1).
\end{eqnarray}
The boundary conditions are taken open, as in the main text, and fixed by vector $\ket{0}$ and the respective dual. Obviously, this is a 
representation of the product  $\ket{0,\dots, 0}$ with norm $\braket{\psi}{\psi}=1$.
For $\ket{\phi}$ we choose
\begin{eqnarray}
	B[0]:=
	\text{diag}(1,0),\,
	B[1]:=
	\text{diag}(\eta,1),\,
\end{eqnarray}
with the same boundary conditions, for some $\eta>0$.
It is still straightforward to compute the norm, invoking the transfer operator
\begin{equation}
	\mathbb{E}:= B[0]\otimes B^\ast[0]+ B[1]\otimes B^\ast[1] = 
	\text{diag}(1 + \eta^2,  \eta, \eta,1).
\end{equation} 
This gives
\begin{equation}
	\braket{\phi}{\phi}= \langle 0|\mathbb{E}^N\ket{0}=  (1+\eta^2)^N.
\end{equation}
Clearly, for the two states to feature norms that are the same up to a constant, 
an in $N$ exponentially small $\eta>0$ is required. In fact, even for a bond dimension $D=1$ one could have
come to a similar conclusion. However, $\ket{\psi}$ and $\ket{\phi}$ are even vastly different in their entanglement properties, the latter
featuring an entanglement entropy of a symmetrically bisected chain that is extensive in $N$.
 \end{document}